\documentclass[12pt,reqno]{amsart}

\newcommand\version{June 19, 2017}

\usepackage{amsmath,amsfonts,amsthm,amssymb,amsxtra}

\newtheorem{theorem}{Theorem}
\newtheorem{proposition}[theorem]{Proposition}
\newtheorem{lemma}[theorem]{Lemma}
\newtheorem{corollary}[theorem]{Corollary}

\theoremstyle{definition}

\theoremstyle{remark}

\newtheorem{remark}[theorem]{Remark}

\newcommand{\C}{\mathbb{C}}

\renewcommand{\epsilon}{\varepsilon}

\newcommand{\N}{\mathbb{N}}

\renewcommand{\phi}{\varphi}
\newcommand{\R}{\mathbb{R}}

\DeclareMathOperator{\diag}{diag}

\DeclareMathOperator{\re}{Re}

\DeclareMathOperator{\tr}{Tr}

\begin{document}

\title[Norms of quantum Gaussian multi-mode channels --- \version]{Norms of quantum Gaussian multi-mode channels}

\author{Rupert L. Frank}
\address[Rupert L. Frank]{Mathematisches Institut, Ludwig-Maximilans Universit\"at M\"unchen, Theresienstr. 39, 80333 M\"unchen, Germany, and Department of Mathematics, California Institute of Technology, Pasadena, CA 91125, USA}
\email{rlfrank@caltech.edu}

\author{Elliott H. Lieb}
\address[Elliott H. Lieb]{Departments of Mathematics and Physics, Princeton University, Princeton, NJ 08544, USA}
\email{lieb@princeton.edu}

\begin{abstract}
We compute the $\mathcal S^p \to \mathcal S^p$ norm of a general Gaussian gauge-covariant multi-mode channel for any $1\leq p<\infty$, where $\mathcal S^p$ is a Schatten space. As a consequence, we verify the Gaussian optimizer conjecture and the multiplicativity conjecture in these cases. 
\end{abstract}

\maketitle

\renewcommand{\thefootnote}{${}$} \footnotetext{\copyright\, 2017 by
  the authors. This paper may be reproduced, in its entirety, for
  non-commercial purposes.}

\section{Introduction}

Gaussian quantum channels play a fundamental role in quantum information theory and quantum optics. They appear, for instance, as a model of attenuation, amplification and noise in electromagnetic communications through metal wires, optical fibers or free space. Despite their ubiquity several fundamental mathematical questions about their structure remain still unsolved. Among them are the Gaussian optimizer conjecture and the additivity conjecture. Our goal here is to contribute a new family of special cases in which we can verify both of these conjectures.

The two conjectures are concerned with the norm of a Gaussian channel acting from a Schatten space $\mathcal S^p$ to a Schatten space $\mathcal S^q$. (We recall the definition of Schatten spaces at the beginning of the following section.) The Gaussian optimizer conjecture states that, in order to compute this norm, it suffices to test the channel on Gaussian states. An affirmative answer to this question would be a non-commutative analogue of a theorem by one of us (E.H.L.) which says that in order to compute the norm of an integral operator with a Gaussian integral kernel from $L^p(\R^d)$ to $L^q(\R^d)$ it suffices to test the integral operator on Gaussian functions \cite{Li}. The Gaussian optimizer conjecture is known to be true for gauge-covariant multi-mode Gaussian channels if $p=1$ \cite{MaGiHo,Ho1} (see also \cite{GiHoGP} for the proof of the entropy version) and for a subclass of gauge-covariant single-mode channels (namely quantum limited attenuators and amplifiers) for any $p$ and $q$ \cite{DPTrGi} (see also \cite{DPTrGi2} for a proof of the entropy version). Our \emph{main result} (Theorem \ref{maingauss}) is that the Gaussian optimizer conjecture is true for gauge-covariant multi-mode channels if $p=q$. Moreover, we are able to compute the corresponding norm explicitly in terms of the parameters of the channel.

The additivity conjecture asks whether the $\mathcal S^p\to\mathcal S^q$ norm of an $M$-fold tensor product of a Gaussian channel is equal to the $M$-th power of the norm of the channel (so the logarithms of the norms are additive, explaining the name of the conjecture). For a history of this problem and a review of some important results we refer to \cite{Ho1} and the references therein. For general quantum channels this additivity is known to be false, but it has been suggested that it might be true for Gaussian channels. Again the conjecture has been verified for gauge-covariant multi-mode channels if $p=1$ \cite{Ho1}. As a \emph{consequence of our main result} (Corollary \ref{cormult}), we are able to conclude that the additivity conjecture holds for $p=q$ for general gauge-covariant multi-mode channels.

The main ingredient in our proof is an abstract bound on the $\mathcal S^p\to\mathcal S^p$ norm of a positive (not necessarily completely positive and not necessarily trace preserving) map on operators (Theorem \ref{mainabstract}). This bound is strongly motivated by the works \cite{Be} and \cite{MHRe} and is obtained by a simple complex interpolation argument. What is remarkable is that this bound is optimal for gauge-covariant Gaussian channels. This is verified in the proof of Theorem \ref{maingauss} using explicit computations with Gaussian states. We will show there that the norm is attained asymptotically in the limit of an infinite temperature thermal state. Note that this is in contrast to the case $p=1$ where the norm is attained at the vacuum (which corresponds to zero temperature). Also, our explicit expression for the $\mathcal S^p\to\mathcal S^p$ norm shows that it is completely determined by the amplification/attenuation matrix $K^*K$ characteristic of the channel, whereas the explicit expression for the $\mathcal S^1\to\mathcal S^p$ norm \cite[Subsection 3.5]{Ho1} shows that the latter is determined by the noise matrix $\mu-K^*K/2$ of the channel. Therefore, our results are in some sense complementary to those in \cite{GiHoGP,Ho1,MaGiHo}, although the mathematical tools are completely different.

\subsection*{Acknowledgements}

We are grateful to Mark Wilde for references concerning Corollary~\ref{entgain} and to Saikat Guha for discussions about Gaussian channels. Partial support by the U.S. National Science Foundation through grants DMS-1363432 (R.L.F.) and PHY-1265118 (E.H.L.) is acknowledged.


\section{An abstract norm bound}

In this section we present a bound in the general setting of a separable complex Hilbert space $\mathfrak H$. We denote by $\mathcal B$ the bounded operators on $\mathfrak H$ and by $\mathcal S^p$, $1\leq p<\infty$, the Schatten class operators of order $p$, that is, the space of all compact operators for which
$$
\| K \|_{\mathcal S^p} = \left( \tr (K^*K)^{p/2} \right)^{1/p} <\infty \,.
$$
As usual, we set $p'=p/(p-1)$ and, given a linear map $\mathcal N:\mathcal S^1\to\mathcal S^1$, we denote the dual map by $\mathcal N^*$.

\begin{theorem}\label{mainabstract}
Let $\mathcal N:\mathcal B\to\mathcal B$ be positive. Then for any $1<p<\infty$,
$$
\left\| \mathcal N \right\|_{\mathcal S^p \to \mathcal S^p} \leq \left\| \mathcal N(1) \right\|_{\mathcal B}^{1/p'}  \left\| \mathcal N^*(1) \right\|_{\mathcal B}^{1/p} \,.
$$
\end{theorem}

We emphasize that we only assume positivity of $\mathcal N$, not complete positivity.

This theorem is an immediate consequence of the following two lemmas.

\begin{lemma}
Let $\mathcal N:\mathcal S^1\to\mathcal S^1$. Then for any $1< p<\infty$,
$$
\left\| \mathcal N \right\|_{\mathcal S^p \to \mathcal S^p} \leq \left\| \mathcal N \right\|_{\mathcal B\to\mathcal B}^{1/p'}  \left\| \mathcal N \right\|_{\mathcal S^1\to\mathcal S^1}^{1/p}  \,.
$$
\end{lemma}

This lemma, at least in the finite dimensional case, is a special case of a result of \cite{Be}. We include the proof for the sake of completeness.

\begin{proof}
We may assume that $\left\| \mathcal N \right\|_{\mathcal B\to\mathcal B}<\infty$, for otherwise there is nothing to prove. Let $X\in\mathcal S^p$ and write $X=U|X|$ with a partial isometry $U$ and $|X|=(X^* X)^{1/2}$. Moreover, let $K$ be a finite rank operator and write $K=V|K|$ with a partial isometry $V$. The function
$$
f(z) := \tr V |K|^{p'(1-z)} \mathcal N(U |X|^{pz})
$$
is analytic in $\{0< \re z<1\}$ and continuous up to the boundary. Moreover, we have for $y\in\R$,
$$
|f(iy)| = \left| \tr V |K|^{p'(1-iy)} \mathcal N(U |X|^{ipy}) \right| \leq  \left\| \mathcal N(U |X|^{ipy} ) \right\|_{\mathcal B} \tr |K|^{p'} \leq \left\| \mathcal N \right\|_{\mathcal B\to\mathcal B} \tr |K|^{p'} \,.
$$
and
\begin{align*}
|f(1+iy)| & = \left| \tr V |K|^{-ip'y} \mathcal N(U|X|^{p(1+iy)}) \right| \leq \left\| \mathcal N(U|X|^{p(1+iy)}) \right\|_{\mathcal S^1} \\
& \leq \left\| \mathcal N \right\|_{\mathcal S^1\to\mathcal S^1} \left\|  U|X|^{p(1+iy)} \right\|_{\mathcal S^1} = \left\| \mathcal N \right\|_{\mathcal S^1\to\mathcal S^1} \tr |X|^p \,.
\end{align*}
We conclude from Hadamard's three line lemma (see, e.g., \cite[Thm. 5.2.1]{Si}) that
$$
\left| \tr K \mathcal N(X) \right| = |f(1/p)| \leq \left( \left\| \mathcal N \right\|_{\mathcal B\to\mathcal B} \tr |K|^{p'} \right)^{1/p'} \left( \left\| \mathcal N \right\|_{\mathcal S^1\to\mathcal S^1} \tr |X|^p \right)^{1/p}.
$$
By duality and density of finite rank operators we conclude that
$$
\left\| \mathcal N(X) \right\|_p \leq \left\| \mathcal N \right\|_{\mathcal B\to\mathcal B}^{1/p'} \left( \left\| \mathcal N \right\|_{\mathcal S^1\to\mathcal S^1} \tr |X|^p \right)^{1/p} \,.
$$
This is the claimed bound.
\end{proof}

\begin{lemma}
Let $\mathcal N:\mathcal S^1\to\mathcal S^1$ be positive. Then
$$
\left\| \mathcal N\right\|_{\mathcal B\to\mathcal B} = \left\|\mathcal N(1)\right\|_{\mathcal B}
$$
and
$$
\left\| \mathcal N\right\|_{\mathcal S^1\to\mathcal S^1} = \left\|\mathcal N^*(1)\right\|_{\mathcal B} \,.
$$
\end{lemma}

Our proof of this lemma is based on the Russo--Dye theorem and has some similarity with an argument in \cite{MHRe}.

\begin{proof}
We recall that, as a consequence of the Russo--Dye theorem \cite{RuDy} (which says that operators with norm one can be approximated in norm by convex combinations of unitary operators), one has
$$
\left\| \mathcal N\right\|_{\mathcal B\to\mathcal B} = \sup_{U} \|\mathcal N(U)\|_{\mathcal B} \,,
$$
where the supremum is over unitaries. (This is true even without the positivity assumption on $\mathcal N$.) We now show that for positive $\mathcal N$ and any unitary $U$ one has $\|\mathcal N(U)\|_{\mathcal B} \leq \|\mathcal N(1)\|_{\mathcal B}$, which proves the first part of the lemma.

By the spectral theorem for unitary operators, we have
$$
U = \int_{[-\pi,\pi)} e^{i\theta} dE_U(\theta)
$$
where $dE_U$ is a positive operator valued measure on $[-\pi,\pi)$ with
$$
\int_{[-\pi,\pi)} dE_U(\theta) = 1\,.
$$
For $\phi,\psi\in\mathfrak H$ we have
$$
\langle\phi,\mathcal N(U)\psi\rangle = \int_{[-\pi,\pi)} e^{i\theta} \langle\phi,\mathcal N(dE_U(\theta))\psi\rangle \,.
$$
Since the measure is positive and $\mathcal N$ is positive, we have
\begin{align*}
\left| \langle\phi,\mathcal N(U)\psi\rangle \right| 
& \leq \left( \int_{[-\pi,\pi)} \langle\phi,\mathcal N (dE_U(\theta))\phi\rangle \right)^{1/2} \left( \int_{[-\pi,\pi)} \langle\psi,\mathcal N(dE_U(\theta))\psi\rangle \right)^{1/2} \\
& = \left( \langle \phi,\mathcal N(1)\phi\rangle\right)^{1/2} \left( \langle\psi,\mathcal N(1)\psi\rangle\right)^{1/2} = \left\|\mathcal N(1) \right\|_{\mathcal B} \|\phi\| \|\psi\| \,.
\end{align*}
This completes the proof of the first part of the lemma.

The second part follows from the first part by duality. In fact,
\begin{align*}
\left\| \mathcal N\right\|_{\mathcal S^1\to\mathcal S^1}
= \sup \frac{\left| \tr V \mathcal N(X)\right|}{\|V\|_{\mathcal B} \|X\|_{\mathcal S^1}} = \sup \frac{\left| \tr \left( \mathcal N^*(V^*) \right)^* X\right|}{\|V^*\|_{\mathcal B} \|X\|_{\mathcal S^1}} = \left\| \mathcal N^*\right\|_{\mathcal B\to\mathcal B} \,,
\end{align*}
and by the first part the right side is equal to $\|\mathcal N^*(1)\|_{\mathcal B}$.
\end{proof}


\section{Application to Gaussian multi-mode channels}

Let $s\in\N$ be the number of modes and let $\mathfrak H$ be the bosonic Fock space over $\C^s$. We denote by $a_1,\ldots, a_s$ and $a^*_1,\ldots,a^*_s$ the usual annihilation and creation operators satisfying $[a_j,a_k^*]=\delta_{jk}$ for $1\leq j,k\leq s$. Moreover, for $z\in\C^s$ let
$$
D(z) = \exp\left( \sum_{j=1}^s \left( z_j a_j^* - \overline z_j a_j\right) \right)
$$
be the displacement (or Weyl) operator.

Let $K$ and $\mu$ be (complex) $s\times s$ matrices with $\mu$ Hermitian and
\begin{equation}
\label{eq:mukineq}
\mu\geq \frac12 (1 - K^* K)
\qquad\text{and}\qquad
\mu\geq -\frac12( 1- K^*K) \,.
\end{equation}
A \emph{gauge-covariant Gaussian $s$-mode channel} $\Phi$ with parameters $K$ and $\mu$ is the linear map $\Phi:\mathcal S^1 \to\mathcal S^1$ which is uniquely determined by
\begin{equation}
\label{eq:gcgc}
\Phi^*( D(z) ) = e^{-z^*\mu z} D(K z)
\qquad\text{for all}\ z\in\C^s \,.
\end{equation}
(We note that here we use the notational convention from \cite{Ho1}, and not that from \cite{GiHoGP}, where $K$ is replaced by $K^*$.) By taking $z=0$ we see that $\Phi$ is trace preserving. Moreover, it is well-known \cite[Prop.~12.31]{Ho} that conditions \eqref{eq:mukineq} are necessary and sufficient for $\Phi$ to be completely positive.

Before stating our main result, let us mention some examples in the single-mode case $s=1$ (so $K$ is a complex number and $\mu$ a real number satisfying $\mu\geq |1-|K|^2|/2$). If $0<K<1$ or $K>1$, then $\Phi$ is the attenuator or amplifier channel, respectively, and equality $\mu=|1-K^2|/2$ corresponds to the quantum limited case. If $K=1$, then $\Phi$ is the additive classical Gaussian noise channel. Important examples of multi-mode channels are given by tensor products of single mode channels, but of course there are multi-mode channels that are not obtained in this way.

\begin{theorem}\label{maingauss}
Let $\Phi$ be a gauge-covariant $s$-mode channel with parameters $K$ and $\mu$ and let $1<p<\infty$. Then $\Phi$ extends to a bounded map from $\mathcal S^p$ to $\mathcal S^p$ if and only if $K$ is invertible, and in this case
$$
\left\|\Phi\right\|_{\mathcal S^p\to\mathcal S^p} = (\det K^*K)^{-1/p'} \,.
$$
\end{theorem}

Before proving this theorem we deduce two simple corollaries. The first one concerns an entropy inequality which gives the minimal entropy gain of a Gaussian gauge-covariant channel. This inequality was previously derived in \cite{Ho0} (even for not necessarily gauge-covariant Gaussian channels) by a different method of proof.

\begin{corollary}\label{entgain}
Let $\Phi$ be a gauge-covariant $s$-mode channel with parameters $K$ and $\mu$ and assume that $K$ is invertible. Then for any non-negative $X$ on $\mathfrak{H}$,
$$
-\tr\Phi(X)\ln\Phi(X) \geq -\tr X\ln X + \left( \ln\det K^* K \right) \tr X \,.
$$
Moreover, this inequality is optimal in the sense that
$$
\inf_{\rho\geq 0\,,\ \tr\rho =1\,,\ -\tr\rho\ln\rho<\infty} \left( -\tr\Phi(\rho)\ln\Phi(\rho) + \tr \rho\ln \rho \right) = \ln\det K^* K \,.
$$
\end{corollary}

The first part of this corollary follows by differentiating the bound $\tr\Phi(X)^p \leq (\det K^*K)^{-p+1}\tr X^p$ from Theorem \ref{maingauss} at the point $p=1$, where it becomes an equality. We comment on the proof of the second part in Remark \ref{entgainopt} below.

The second corollary concerns the multiplicativity problem for Gaussian channels.

\begin{corollary}\label{cormult}
Let $s_1,\ldots,s_M\in\N$ and for each $m=1,\ldots,M$ let $\Phi_m$ be a gauge-covariant $s_m$-mode channel with parameters $K_m$ and $\mu_m$. Then for each $1<p<\infty$,
$$
\left\|\Phi_1\otimes\cdots\otimes\Phi_M \right\|_{\mathcal S^p\to\mathcal S^p} = \left\|\Phi_1\right\|_{\mathcal S^p\to\mathcal S^p} \cdots \left\|\Phi_M\right\|_{\mathcal S^p\to\mathcal S^p} \,.
$$
\end{corollary}

This corollary simply follows from the fact that $\Phi_1\otimes\cdots\otimes\Phi_M$ is a gauge-covariant $(s_1+\ldots+s_M)$-mode channel with parameters $K$ and $\mu$ given as block diagonal matrices with entries $K_m$ and $\mu_m$ and the fact that $\det K^* K = \det K_1^* K_1 \cdots \det K_M^*K_M$.

In order to deduce the upper bound on the norm from Theorem \ref{mainabstract} and to prove a corresponding lower bound we will make use of a computation involving the following family of single-mode Gaussian states parametrized by $E\geq 0$,
\begin{equation}
\label{eq:thermalstates}
\omega_E = \frac{1}{E+1} \sum_{n=0}^\infty \left( \frac{E}{E+1}\right)^{n} |n\rangle\langle n| \,.
\end{equation}
Here $(|n\rangle)_{n=0}^\infty$ is the canonical basis in the single-mode space, i.e., the Fock space over $\C$ which is, of course, simply $\ell^2(\N_0)$ with $\N_0=\{0,1,2,\ldots\}$. (The states $\omega_E$ are thermal states of the Hamiltonian $a^* a$.) Then the $s$-fold tensor product
$$
\omega_E^{\otimes s}
$$
is a Gaussian state on the $s$-mode space $\mathfrak H$.

\begin{lemma}\label{comps}
Let $\Phi$ be a gauge-covariant Gaussian $s$-mode channel with parameters $K$ and $\mu$, and let $E\geq 0$. Then, for $1\leq p<\infty$,
\begin{align}
\label{eq:phigaussian}
\left\|\omega_E^{\otimes s}\right\|_{\mathcal S^p} & = \left( (E+1)^p - E^p \right)^{-s/p} \,, \\
\left\|\Phi(\omega_E^{\otimes s})\right\|_{\mathcal S^p} & = \left( \det \left( \left( (E+1/2)K^*K + \mu +1/2\right)^p - \left( (E+1/2)K^*K + \mu -1/2\right)^p \right) \right)^{-1/p} \notag
\end{align}
and, if $K$ is invertible,
\begin{align}
\label{eq:phi1}
\Phi(1) = (\det K^*K)^{-1} \,.
\end{align}
\end{lemma}

We note that the first inequality in \eqref{eq:mukineq} implies that $(E+1/2)K^*K+\mu-1/2\geq E K^*K\geq 0$, so there is no problem with defining its $p$-th power.

\begin{proof}[Proof of Lemma \ref{comps}]
We denote by $e_1,\ldots,e_s$ the eigenvalues of $(E+1/2)K^*K+\mu-1/2$ and let $U_E$ be a unitary $s\times s$ matrix such that
$$
U_E \left( (E+1/2)K^* K+\mu-1/2 \right) U_E^* = \diag (e_1,\ldots,e_s) \,.
$$
By basic representation theory there is a unitary $V_E$ on $\mathfrak H$ such that
$$
V_E D(\zeta) V_E^* = D(U_E^*\zeta)
\qquad\text{for all}\ \zeta\in\C^s \,.
$$
It is well-known \cite[(12.32)]{Ho} that
\begin{equation}
\label{eq:charfcngauss}
\tr \omega_E D(z) = e^{-(E+1/2)|z|^2} 
\qquad\text{for all}\ z\in\C \,,
\end{equation}
and therefore
$$
\tr \omega_E^{\otimes s} D(z) = e^{-(E+1/2)|z|^2} 
\qquad\text{for all}\ z\in\C^s \,.
$$
Thus, by \eqref{eq:gcgc}
\begin{align*}
\tr V_E^* \Phi(\omega_E^{\otimes s}) V_E D(\zeta) & = \tr \Phi(\omega_E^{\otimes s})D(U_E^*\zeta) = e^{-(U_E^*\zeta)^*\left( \mu + (E+1/2)K^* K\right) U_E^*\zeta} \\
& = \prod_{j=1}^s e^{-(e_j+1/2)|\zeta_j|^2}.
\end{align*}
According to \eqref{eq:charfcngauss} the right side is $\prod_{j=1}^s \tr \omega_{e_j} D(\zeta_j) = \tr (\omega_{e_1}\otimes \cdots \otimes \omega_{e_s}) D(\zeta)$. Since Gaussian channels maps Gaussian states into Gaussian states \cite[Sec.~12.4]{Ho} and since Gaussian states are uniquely determined by their characteristic function \cite[Thm.~12.17]{Ho}, we conclude that 
$$
V_E^* \Phi(\omega_E^{\otimes s}) V_E = \omega_{e_1}\otimes \cdots \otimes \omega_{e_s} \,.
$$
Since $V_E$ is unitary, we infer
$$
\left\| \Phi(\omega_E^{\otimes s}) \right\|_{\mathcal S^p}^p = \prod_{j=1}^s \left\| \omega_{e_j} \right\|_{\mathcal S^p}^p \,.
$$
Thus, for the proof of both statements in \eqref{eq:phigaussian} it suffices to compute $\|\omega_E\|_{\mathcal S^p}$. By the explicit expression, we have
$$
\left\| \omega_E \right\|_{\mathcal S^p}^p = \frac{1}{(E+1)^p} \sum_{n=0}^\infty \left( \frac{E}{E+1}\right)^{np} = \frac{1}{(E+1)^p} \ \frac{1}{1- \left(\frac{E}{E+1}\right)^p} = \frac{1}{(E+1)^p - E^p} \,.
$$
This leads to the claimed expressions for the Schatten norms.

It remains to prove \eqref{eq:phi1} under the assumption that $K$ is invertible. It follows from perturbation theory that the eigenvalues of $E^{-1} \left( (E+1/2)KK^*+\mu-1/2 \right)$ converge to those of $K^* K$ as $E\to\infty$ and that one can choose the unitaries $U_E$ in such a way that they converge to a unitary $U_\infty$ on $\C^s$ such that
$$
U_\infty K^* K U_\infty^* = \diag(\kappa_1^2,\ldots,\kappa_s^2)
$$
for some $\kappa_j> 0$, $j=1,\ldots,s$. (The fact that $\kappa_j\neq 0$ comes from the assumed invertibility of $K$.) This implies that the corresponding $V_E$ converge in norm to a unitary $V_\infty$ on $\mathfrak H$ such that
$$
V_\infty D(\zeta) V_\infty^* = D(U_\infty^* \zeta)
\qquad\text{for all}\ \zeta\in\C^s \,.
$$
Let $\Psi\in\mathfrak H$. Since $(E+1)\omega_E$ is increasing with respect to $E$ and converges weakly to the identity, we see that
$$
(E+1)^s \langle\Psi|\Phi(\omega_E^{\otimes s})|\Psi\rangle
$$
is increasing with respect to $E$ and its limit, if it is finite, coincides necessarily with $\langle\Psi|\Phi(1)|\Psi\rangle$. On the other hand, according to the preceeding computation, we have
$$
\langle\Psi|\Phi(\omega_E^{\otimes s})|\Psi\rangle = \tr V_E^* |\Psi\rangle\langle\Psi| V_E \left( \omega_{e_1}\otimes\cdots\otimes\omega_{e_s}\right) \,.
$$
The operators $V_E^*|\Psi\rangle\langle\Psi| V_E$ are compact and converge in norm to $V_\infty^* |\Psi\rangle\langle\Psi| V_\infty$. Therefore, since $e_j\to\infty$ as $E\to\infty$,
$$
(e_1+1)\cdots (e_s+1) \tr V_E^* |\Psi\rangle\langle\Psi| V_E \left( \omega_{e_1}\otimes\cdots\otimes\omega_{e_s}\right) \to \tr V_\infty^* |\Psi\rangle\langle\Psi| V_\infty = \|\Psi\|^2 \,.
$$
Thus, we conclude that
$$
\lim_{E\to\infty} (E+1)^s \langle\Psi|\Phi(\omega_E^{\otimes s})|\Psi\rangle = \lim_{E\to\infty} \frac{(E+1)^s}{(e_1+1)\cdots (e_s+1)} \|\Psi\|^2 \,.
$$
According to the discussion before, we have $e_j/E \to\kappa_j^2$ for $j=1,\ldots, s$ and therefore
$$
\lim_{E\to\infty} \frac{(E+1)^s}{(e_1+1)\cdots (e_s+1)} = \frac{1}{\kappa_1^2\cdots\kappa_s^2} = \frac{1}{\det K^* K} \,.
$$
This completes the proof of \eqref{eq:phi1}.
\end{proof}

\begin{proof}[Proof of Theorem \ref{maingauss}]
\emph{Upper bound.} Since $\Phi$ is trace preserving, we have $\Phi^*(1)=1$. Moreover, by \eqref{eq:phi1} we have $\Phi(1) = (\det K^*K)^{-1}$, provided the latter is finite. Inserting this into the bound from Theorem \ref{mainabstract} we obtain
$$
\left\|\Phi\right\|_{\mathcal S^p\to\mathcal S^p} \leq (\det K^*K)^{-1/p'} \,.
$$

\emph{Lower bound.} According to \eqref{eq:phigaussian} we have
$$
\frac{\left\|\Phi(\omega_E^{\otimes s})\right\|_{\mathcal S^p}^p}{\left\|\omega_E^{\otimes s}\right\|_{\mathcal S^p}^p} = \prod_{j=1}^s \frac{(E+1)^p-E^p}{(e_j+1)^p-e_j^p} \,,
$$
where $e_j$ are the eigenvalues of $(E+1/2)K^*K+\mu-1/2$. As in the previous proof, we have $e_j/E\to\kappa_j^2$, where $\kappa_j^2$ are the eigenvalues of $K^*K$. This yields
$$
\lim_{E\to\infty} \frac{(E+1)^p-E^p}{(e_j+1)^p-e_j^p} = \frac{1}{\kappa_j^{2(p-1)}}
$$
in the sense that the left side diverges to $+\infty$ if $\kappa_j=0$. This proves that
$$
\lim_{E\to\infty} \frac{\left\|\Phi(\omega_E^{\otimes s})\right\|_{\mathcal S^p}^p}{\left\|\omega_E^{\otimes s}\right\|_{\mathcal S^p}^p} = \prod_{j=1}^s \frac1{\kappa_j^{2(p-1)}} = \frac{1}{(\det K^*K)^{p-1}}
$$
in the sense the the left side diverges to $+\infty$ if $K$ is not invertible. Since the left side is a lower bound on $\|\Phi\|_{\mathcal S^p\to\mathcal S^p}^p$, we conclude that the upper bound in the theorem is best possible.
\end{proof}

\begin{remark}
\label{entgainopt}
The optimality statement in Corollary \ref{entgain} is shown similarly as in the preceeding statement. In fact, one verifies that
$$
-\tr \Phi(\omega_E^{\otimes s})\ln\Phi(\omega_E^{\otimes s}) + \tr \omega_E^{\otimes s} \ln \omega_E^{\otimes s} \to \ln\det K^*K
\qquad\text{as}\ E\to\infty \,.
$$
\end{remark}

We end this paper with a result about the $\mathcal S^p\to\mathcal S^q$ norm for $q<p$. This generalizes a result of \cite{DPTrGi} for quantum-limited single mode channels.

\begin{proposition}
Let $\Phi$ be a gauge-covariant $s$-mode channel with parameters $K$ and $\mu$ and let $1\leq q<p<\infty$. Then $\Phi$ does not extend to a bounded map from $\mathcal S^p$ to $\mathcal S^q$.
\end{proposition}

This proposition follows by the same computations as in the proof of the lower bound in Theorem \ref{maingauss} using the same family of trial states and letting $E\to\infty$.



\bibliographystyle{amsalpha}

\end{document}